\documentclass[conference,letterpaper]{IEEEtran}

\usepackage[lmargin=0.75in,rmargin=0.75in,bmargin=0.75in,tmargin=0.75in]{geometry}

\usepackage{amsmath,amsthm,amssymb}
\allowdisplaybreaks[2]
\usepackage{mathtools}

\usepackage{microtype}

\usepackage{changes}
\usepackage[pdftex]{graphicx}
\newtheorem*{remark}{Remark}

\theoremstyle{plain}
\newtheorem{theorem}{Theorem}[section]
\newtheorem{lemma}[theorem]{Lemma}

\newtheorem{corollary}[theorem]{Corollary}

\theoremstyle{definition}

\theoremstyle{remark}

\newtheorem{assumption}[theorem]{Assumption}
\newcommand{\mX}{\mathcal{X}}
\newcommand{\mN}{\mathcal{N}}

\newcommand{\bbR}{\mathbb{R}}
\newcommand{\mF}{\mathcal{F}}

\title{Mechanism Design for Fair Allocation}

\author{\IEEEauthorblockN{Abhinav Sinha and
		Achilleas Anastasopoulos}
	\IEEEauthorblockA{EECS Department,
		University of Michigan\\
		Email: \{absi,anastas\}@umich.edu}
}

\begin{document}
	%\begin{spacing}{1.1}
	%\setlength{\jot}{1ex}
	
	\newgeometry{top=1in,bmargin=0.75in,rmargin=0.75in,lmargin=0.75in}
	
	\maketitle 
	
	\begin{abstract}
		Mechanism design for a social utility being the sum of agents' utilities (SoU)
		is a well-studied problem.  There are, however, a number of problems of
		theoretical and practical interest where a designer may have a different
		objective than maximization of the SoU. One motivation for this is the desire
		for more equitable allocation of resources among agents. A second, more subtle,
		motivation is the fact that a fairer allocation indirectly implies less
		variation in taxes which can be desirable in a situation where (implicit)
		individual agent budgetary constraints make payment of large taxes unrealistic.
		In this paper we study a family of social utilities that provide fair
		allocation (with SoU being subsumed as an extreme case) and derive conditions
		under which
		Bayesian and Dominant strategy implementation is possible.
		Furthermore, it is shown how a simple modification of the above mechanism can guarantee full Bayesian implementation. 
		Through a numerical example it is shown that the proposed method can result in
		significant gains both in allocation fairness and tax reduction.
	\end{abstract}

	\section{Introduction} 
	\label{secintro}

	Mechanism design is a well-established framework for dealing with decentralized resource allocation problems in the presence of strategic agents.
	The corresponding literature is vast, especially in the domain of Dominant, Nash and Bayesian implementation, \cite{hurwicz2006,maskin2002implementation,jackson,krishna}. 
	In the area of Engineering--and in particular in the area of Networks--
	a majority of the works start with the assumption that the social objective is the sum of individual utilities (SoU) of all the system's agents.
	
	%==========================================
	
	There are strong mathematical reasons for preferring SoU as the resource allocation objective. The main one being that for SoU, in conjunction with quasi-linear utilities, agents' individual goals can be aligned directly with the overall social objective, so that when agents maximize their net utility they are simultaneously maximizing the overall social objective. The VCG mechanism~\cite{groves1973incentives} is the most prominent example of this. With social objective $ \sum_{i=1}^N v(\hat{x}(\theta);\theta_i) $ and individual utility of $ v(\hat{x}(\phi);\theta_i) - t_i(\phi) $, VCG taxes $ t_i(\phi) = - \sum_{j \ne i} v(\hat{x}(\phi);\phi_j) + f_i(\phi_{-i}) $ precisely ask the user to perform social objective maximization whilst maximizing self utility.
	
	%==========================================
	
	There are, however, a number of problems of theoretical and practical interest where a designer may have a different objective than maximization of the SoU.
	One obvious reason for such a preference is the desire of the social planner to introduce fairness in the allocation process. Consider for example, a network where optimizing SoU results in one (or a few) agents receiving almost all the available resources and everyone else receiving an appreciably lower portion.
	In instances of this kind, appealing to fairness or equality, a system designer may genuinely want allocations which are more equitable, even if this may come at the cost of reduced revenue.
	
	%==========================================
	
	A second--more subtle--reason for wanting a different social objective is the fact that the standard mechanism design framework does not provide any formal way of limiting the range of the taxes/subsidies required at equilibrium. This implies that when strong budget balance is imposed on the mechanism, the magnitude of the monetary transfers (taxes/subsidies) can vary greatly among agents (with those who benefit more from the allocation having to contribute more as well). This can be a significant practical problem since it does not take into account the budget constraints of individual agents.
	The work of~\cite{wei2014competitive} is attacking this problem by
	considering the dual version of the resource allocation problem and putting additional structure on the dual variables to ensure less variation between them. This would typically ensure less fluctuation in prices (and thus in the monetary transfers) since it is well-known that the dual variables for resource allocation optimization problems act as prices in the corresponding markets.
	A modified social objective provides an alternative way for dealing with the issue of large tax variation. By making the social objective a more concave function of the utilities (compared to the SoU) a smaller variation of the taxes is expected.
	
	%===================================================
	
	As soon as one moves away from the SoU objective, due to the social and individual objectives not aligning with each other, basic design techniques like VCG mechanism are not useful.
	This is one of the main reasons why there are significantly fewer results on fairness in the mechanism design literature.
	%%%%===================================================
	In~\cite{cole13}, authors use the concept of ``proportional fairness'' (similar to the sum of $ \log $ of utilities) to reduce disparity. The focus is on a tax-less mechanism where the contract proposes to throw away existing resources (``resource-burning'') in order to tax untruthful agents. This mechanism achieves at least $ \frac{1}{e} $ fraction of proportionally fair allocation. In~\cite{vohra2012}, optimal auctions in the Bayesian set up are derived, such that instead of efficiency maximization or revenue maximization, a linearly combined metric is maximized which favors exchange of goods at low prices (thereby ensuring fairer trade).
	``Envy-freeness'' is another well-known criterion for equitable allocations (see~\cite{varian1974equity}). This notion was originally proposed for exchange economies where an allocation is called envy-free if no agent is strictly better off by taking  \restoregeometry someone else's allocation instead of their own\footnote{Note that this exchange refers to both allocation of good and taxes.}. Such a notion was argued in terms of the stability it provides, since each agent may be content with what they have comparing to the possible option of acquiring someone else's allocation. This, however, is an ex-post notion and in general imposes quite stringent constraints on the design. For a large enough environment it may indeed be impossible to achieve envy-freeness in optimal allocation.

	It is interesting to note that the problem of mechanism design for risk-averse agents (see for e.g.,~\cite{maskin1984optimal,peres_risk}) is in some respects the opposite problem of the one addressed here.
	In that case, the social planner's objective is relatively more ``aggressive'' compared to the more risk averse
	individual objectives.
	
	%==============================================
	
	In this paper, we ask if and how we can design mechanisms that implement social objectives that are especially designed for fairness and go beyond the standard paradigm of SoU.
	We seek a methodology that is flexible enough to create space for the designer when envy-freeness may not be feasible.
	We concentrate on the form of the social objective given by the additive function $\sum_{i=1}^N g_{\epsilon}(v(x;\theta_i))$ where $v(x;\theta_i)$ is the utility of the $ i $-th user with allocation $ x $ and type $\theta_i$. Here we take $ g_\epsilon(z) = z - \epsilon f(z) $ as a family of concave functions parameterized by $ \epsilon > 0 $, with $f(\cdot)$ an arbitrary convex function.
	In this setup, the SoU is a special case with $\epsilon=0$, while as the parameter $\epsilon$
	increases, more fairness is built into the allocation.
	The form $ g_\epsilon(z) = z - \epsilon f(z) $ is considered without significant loss of generality, since it closely emulates Taylor's series (w.r.t. $ \epsilon $) of many interesting families of concave functions. Consider for example the family $ h_\epsilon(z) = z^{1-\epsilon} $, the Taylor's series for this family is $ {h}_\epsilon(z) = z - \epsilon (z \log(z)) + \mathit{o}(\epsilon) $, where $ z\log(z) $ is indeed convex. Note that unlike agents' utilities, the choice of the specific function $ g_\epsilon $ is in the designer's hand, as long as it serves the design objective of fairer allocation. 
	Furthermore, in Section~\ref{secdisc} we discuss how the results are still valid even when $ f(z) $ depends on $ \epsilon $ (under certain conditions).

	Within this framework we ask for which values of the parameter $\epsilon$ and under what conditions for the convex function $ f(\cdot)$ is Dominant strategy implementation possible and when is Bayesian Nash Equilbrium (BNE) implementation possible.
	We show that indeed mechanism design is possible provided $\epsilon $ is not too large, by providing an upper bound on the range of $\epsilon$.
	This is done by formulating the incentive compatibility constraints as a set of linear inequalities on the design variables and checking whether this system (together with strong budget balance and/or individual rationality) is feasible. Our proving techniques follow closely the work of~\cite{dAGV90, dAGV03}. Not surprisingly, the results in the Bayesian set-up are derived under certain assumptions on the prior beliefs, $p_i(\theta_{-i}|\theta_i)$, of agents, which are trivially satisfied for the case where $p_i(\theta_{-i}|\theta_i)=p_i(\theta_{-i})$.

	For the case of Bayesian mechanism design, we also propose a modification to our mechanism which ensures not only that truth-telling is a BNE but also that it is the only BNE. This modification doesn't make significant changes to the mechanism and only requires exchange of one additional message.

	We finally demonstrate (through a numerical example) that the range of $ \epsilon $ is sufficient to provide quite significant gains in fairness - as measured by the decrease in Gini Coefficient of the utilities. In addition, and this relates to the second reason we mentioned above regarding our motivation for this work, the results show a significant decrease in the variance of required taxes to achieve incentive compatibility.

	%===============================================
	
	The remaining of this paper is organized as follows: Section~\ref{seccp} defines the Centralized problem which has been modified for fairness. The next two sections prove the existence of mechanisms that implement the aforementioned social objectives for type sets of size two in Dominant strategy (Section~\ref{secds}) and for general type sets in BNE (Section~\ref{secbne}). Section~\ref{secfbi} presents the modification for full Bayesian implementation.
	The numerical example is presented in Section~\ref{secexam}. Finally, Section~\ref{secdisc} discusses future work and immediate extensions of the results in here.

	\section{Centralized Problem} \label{seccp}
	
	For a system of agents $ \mN = \{ 1,\ldots,N \} $, efficient allocation is calculated via the following optimization problem:
	\begin{gather} \label{eqcp}
	\hat{x}_\epsilon(\theta) = \arg\max_{x \in \mX} \sum_{i \in \mN} g_\epsilon(v(x;\theta_i)),
	\end{gather}
	where
	$ \mX \subset \bbR_+^N $ is the constraint set,
	$ \theta = (\theta_i)_{i \in \mN} \in \Theta \triangleq \times_{i \in \mN} \Theta_i $ is the \emph{type profile} of agents and $ \Theta_i $ is the discrete type set for agent $i$ with $ \vert \Theta_i \vert = L_i $. The utility function $ v(x;\theta_i) $ measures agent's $i$ satisfaction at allocation $ x $ with private type being $ \theta_i $. A concave transformation $ g_\epsilon(\cdot) $ is applied for making the allocation \emph{fairer} compared to the SoU setup.
	It is further assumed that
	the functions $ g_\epsilon:\bbR \rightarrow \bbR $ and $ v(\cdot;\theta_i):\bbR^N_+ \rightarrow \bbR $ are such that the optimization has a unique solution (e.g., if $ v(\cdot;\theta_i) $ is concave and $ g_\epsilon $ is concave and increasing and $ \mX $ is a convex set).
	Specifically the form $ g_\epsilon(z) = z - \epsilon f(z) $ for $ \epsilon \ge 0 $ is considered here, where $ f(\cdot) $ is assumed to be a bounded convex function. Note that at $ \epsilon = 0 $ optimization~\eqref{eqcp} becomes the SoU problem, while as $ \epsilon $ increases from $ 0 $ the function $ g_\epsilon $ has a stronger concave component.

	Define, $ \forall $ $ i \in \mN $, $ \mF_i \coloneqq \{ (\psi_i,\xi_i) \in \Theta_i^2 \mid \psi_i \ne \xi_i \} $. With this we define the difference, $ \forall~i \in \mN $, $ (\theta_i,\phi_i) \in \mF_i $, $ \theta_{-i} \in \Theta_{-i} $, $ \epsilon > 0 $,
	\begin{multline} \label{EQkdef}
	K_i(\theta_i,\phi_i,\theta_{-i},\epsilon) \coloneqq \sum_{j \in \mN} g_\epsilon(v(\hat{x}(\theta_i,\theta_{-i});\theta_j))
	\\
	- \sum_{j \in \mN} g_\epsilon(v(\hat{x}(\phi_i,\theta_{-i});\theta_j)).
	\end{multline}
	Optimality conditions from~\eqref{eqcp} give that the above difference is always non-negative. However, due to the finite type spaces, the difference above is expected to be strictly positive. We make appropriate assumptions in this regard later.

	\section{Dominant Strategy Implementation} \label{secds}

	The Mechanism Design problem in this section is to find a message space $ \mathcal{M} = \times_{i \in \mN} \mathcal{M}_i $ and allocation, tax functions $ (\tilde{x},t):\mathcal{M} \rightarrow \mX \times \bbR^N $ such that the induced game for agents in $ \mN $ with action space $ \mathcal{M} $ and quasi-linear utilities
	\begin{equation}
	\tilde{u}_i(m;\theta_i) = v(\tilde{x}(m);\theta_i) - t_i(m) \quad \forall~m \in \mathcal{M},~i \in \mN
	\end{equation}
	has a dominant strategy equilibrium\footnote{For any $ \theta $, message $ m^\star \in \mathcal{M} $ is a dominant strategy equilibrium if it satisfies
		$ \tilde{u}_i(m_i^\star,m_{-i};\theta_i) \ge \tilde{u}_i(m_i,m_{-i};\theta_i) \quad \forall~m_i \in \mathcal{M}_i,~\forall~m_{-i} \in \mathcal{M}_{-i},~\forall~i \in \mN $. 
	}
	$ m^\star $ for which $ \tilde{x}(m^\star) = \hat{x}_\epsilon(\theta) $, where $ \theta = (\theta_i)_{i \in \mN} $ is the true type profile. This is known as Dominant Strategy Incentive Compatibility (DSIC).
	
	In general, Dominant strategy implementation is very restrictive (note that the well studied VCG mechanisms are no longer applicable since this is not the maximization of SoU). Following~\cite{dAGV79}, 
	the special case of $ L_i = 2 ~ \forall ~ i \in \mN $ is considered in this section. In particular, $ \Theta_i = \{ \theta_i^H,\theta_i^L \} ~\forall~ i \in \mN $.
	
	The proposed mechanism is a direct mechanism, thus $ \mathcal{M}_i = \Theta_i ~\forall~ i \in \mN $. Agents report their types (possibly untruthfully) $ \phi = (\phi_i)_{i \in \mN} $ and the allocation they receive on the basis of this is the optimal allocation $ \hat{x}_\epsilon(\phi) $ for the quoted type profile $\phi$, where $ \hat{x}_\epsilon(\cdot) $ is defined in~\eqref{eqcp}.

	Assuming that for not participating in the mechanism, an agent receives $ 0 $ utility value (including $ 0 $ tax), the voluntary participation condition for Dominant strategy implementation is $ \forall~(\theta_i,\theta_{-i}) \in \Theta,~i \in \mN$,
	\begin{equation}
	v(\hat{x}_\epsilon(\theta_i,\theta_{-i});\theta_i) - t_i(\theta_i,\theta_{-i}) \ge 0.
	\end{equation}
	This is the \emph{ex-post} version of the individual rationality (IR).
	
	The first contribution of this paper is summarized in the following Theorem.
	
	\begin{theorem} \label{thmDSIC0}
		For any $ \epsilon \ge 0 $, there exist taxes $ \big(t_i(\phi)\big)_{\phi \in \Theta,i \in \mN} $ that satisfy the corresponding DSIC (which implies implementation in Dominant strategies) and IR conditions if and only if, $ \forall~\theta_{-i} \in \Theta_{-i},~i \in \mN $,
		\begin{multline} \label{EQthm1cond}
		v(\hat{x}_\epsilon(\theta_i^H,\theta_{-i});\theta_i^H) 
		{} + {} v(\hat{x}_\epsilon(\theta_i^L,\theta_{-i});\theta_i^L) 
		\\
		{} - {} v(\hat{x}_\epsilon(\theta_i^H,\theta_{-i});\theta_i^L) 
		{} - {} v(\hat{x}_\epsilon(\theta_i^L,\theta_{-i});\theta_i^H) \ge 0.
		\end{multline}
	\end{theorem}
	\begin{proof}
		Please see Appendix~\ref{app1}.
	\end{proof}
	
	Next we state the assumption on~\eqref{eqcp} under which the next result in this section is derived.
	
	\paragraph*{\textbf{Condition (A$_{\textbf{D}}$)}}
	Assume that $ \exists $ $ \epsilon_{{max}} > 0 $ such that for all $ 0 \le \epsilon < \epsilon_{max} $, $ \forall~i \in \mN,~(\theta_i,\phi_i) \in \mF_i,~\theta_{-i} \in \Theta_{-i} $,
	\begin{equation} \label{EQAD}
	K_i(\theta_i,\phi_i,\theta_{-i},\epsilon) + K_i(\phi_i,\theta_i,\theta_{-i},\epsilon) > 0.
	\end{equation}

	Condition (A$_{\text{D}}$) and the ones below in Corollary~\ref{corolDSIC}, \ref{corolBSIC1}, \ref{corolBSIC2} and Condition (A$_{\text{B}}$) can be checked only at $ \epsilon = 0 $. By continuity of the optimization on parameter $ \epsilon $, this will imply that these conditions continue to hold for all $ 0 \le \epsilon < \epsilon_{max} $ for some $ \epsilon_{max} > 0 $.

	\begin{theorem} \label{thmDSIC}
		If Condition (A$_{\text{D}}$) is satisfied then $ \exists $ $ \tilde{\epsilon}_{max} > 0 $ such that for all $ 0 \le \epsilon < \tilde{\epsilon}_{max} $ there exist taxes $ \big(t_i(\phi)\big)_{\phi \in \Theta,i \in \mN} $ that satisfy DSIC (which implies implementation in Dominant strategies) and IR.
	\end{theorem}
	\begin{remark} 
		The relation between $ \epsilon_{max} $ and $ \tilde{\epsilon}_{max} $ is defined in the proof, specifically via the relation in~\eqref{EQd7}.
	\end{remark}
	\begin{proof}
		Please see Appendix~\ref{app2}.
	\end{proof}
	
	Stated below is a corollary to the above result, which uses a stricter condition.

	\begin{corollary} \label{corolDSIC}
		If $ \exists $ $ \epsilon_{max} > 0 $ such that for all $ 0 \le \epsilon < \epsilon_{max} $,
		\begin{equation} \label{EQCR1}
		\min_{i \in \mN} \min_{(\theta_i,\phi_i) \in \mF_i} \min_{\theta_{-i} \in \Theta_{-i}} K_i(\theta_i,\phi_i,\theta_{-i},\epsilon) > 0,
		\end{equation}
		then $ \exists $ $ \tilde{\epsilon}_{max} > 0 $ such that for all $ 0 \le \epsilon < \tilde{\epsilon}_{max} $ there exist taxes that satisfy DSIC and IR.
	\end{corollary}
	\begin{proof}
		As~\eqref{EQCR1} implies~\eqref{EQAD}, the Corollary follows from Theorem~\ref{thmDSIC}. Also, as long as the above assumption holds, the value of $ \tilde{\epsilon}_{max} $ will be same as the one derived in the previous proof, since it is only defined by the expression in~\eqref{EQd7}.
	\end{proof}

	\section{Bayesian Implementation} \label{secbne}
	
	In this section the type sets are of arbitrary size. Dominant strategy implementation is too restrictive for the general scenario, hence the next best reasonable solution concept - Bayesian implementation, is considered.
	
	In a Bayesian set up, agents have a prior distribution on the type profile. For agent $ i $, prior is $ p_i \in \Delta(\Theta) $. For basic regularity assume that the prior gives non-zero probability on all points of $ \Theta $ (this is only a technical condition and the ensuing results can be proved without it as well). These priors are assumed to be \emph{common knowledge} between agents and designer - hence there is no need to introduce second order beliefs over the priors and so on.
	
	The mechanism used here is a direct mechanism with allocation function $ \hat{x}_\epsilon(\cdot) $ (same as before). Given the allocation and tax functions $ (\hat{x}_\epsilon,t): \Theta \rightarrow \mX \times \bbR^N $, the utility function in the Bayesian set up for strategy profile $ \{\sigma_j:\Theta_j \rightarrow \Theta_j\}_{j \in \mN} $ is given by, $ \forall~i \in \mN $,
	\begin{equation}
	\tilde{u}_i(\sigma \mid \theta_i) =  \mathbb{E}_{p_i(\cdot \mid \theta_i)} 
	\big[ v\big( \hat{x}_\epsilon(\sigma(\theta)); \theta_i \big) - t_i\big( \sigma(\theta) \big) \big].
	\end{equation}
	where $ \theta_i $ is the true type of agent $ i $.
	The Bayesian implementation condition--also known as Bayesian Strategy Incentive Compatibility (BSIC)--for the direct mechanism is that the truthful strategy $ \sigma_i^\star(\theta_i) = \theta_i, ~\forall~\theta_i \in \Theta_i,~\forall~i \in \mN $ must be a Bayesian Nash equilibrium (BNE)\footnote{Strategy $ \sigma^\star = \big( \sigma_i^\star:\Theta_i \rightarrow \Theta_i \big)_{i \in \mN} $ is a BNE if $ \forall i \in \mN, ~\forall~\theta_i \in \Theta_i, ~\forall~\sigma_i^\prime: \Theta_i \rightarrow \Theta_i $; $ \tilde{u}_i(\sigma_i^\star,\sigma_{-i}^\star \mid \theta_i) \ge \tilde{u}_i(\sigma_i^\prime,\sigma_{-i}^\star \mid \theta_i) $.} for the induced Bayesian game.
	
	In addition, it is required the tax function to have the \emph{Strong Budget Balance} (SBB) property, i.e.,
	\begin{equation}
	\sum_{i \in \mN} t_i(\psi) = 0, ~~~\forall ~\psi \in \Theta.
	\end{equation}
	
	We restrict attention to optimization~\eqref{eqcp} and priors $ \{p_i(\cdot)\}_{ i \in \mN } $ that satisfy the following conditions.

	\paragraph*{\textbf{Condition (A$_{\textbf{B}}$)}} 
	Assume that $ \exists $ $ \epsilon_{max} > 0 $ such that $ H(\epsilon) > 0 $ for all $ 0 \le \epsilon < \epsilon_{max} $, where
	\begin{subequations}
		\label{EQthm2}
		\begin{gather}
		H(\epsilon) \coloneqq \min_{i \in \mN} \min_{(\theta_i,\phi_i) \in \mF_i} \mathbb{E}_{p_i(\cdot \mid \theta_i)} \big[h_i(\theta_i,\phi_i,\theta_{-i},\epsilon)\big],
		\\
		\nonumber 
		h_i(\theta_i,\phi_i,\theta_{-i},\epsilon)
		\coloneqq
		\\
		\sum_{j \in \mN} \Big(v(\hat{x}_\epsilon(\theta_i,\theta_{-i});\theta_j) - v(\hat{x}_\epsilon(\phi_i,\theta_{-i});\theta_j)\Big).
		\end{gather}
	\end{subequations}

	This is the most general form of the assumption needed about the optimization here; after proving Theorem~\ref{thmBSIC}, corollaries are stated with stricter assumptions than above.

	\paragraph*{\textbf{Condition (B)}}
	Assume that for any non-zero vector $ R:=(R(\psi))_{\psi\in\Theta} $ with $ R(\psi)\in \bbR $,
	there does not exist any
	$ \lambda:= (\lambda_k(\theta_k,\phi_k))_{k \in \mN, (\theta_k,\phi_k) \in \mF_k} $ with
	$\lambda_k(\theta_k,\phi_k) \in \bbR_+  $
	such that $ \forall~i \in \mN $, $ \forall~\psi \in \Theta $,
	\begin{multline}
	p_i(\psi_{-i} \mid \psi_i) \sum_{\substack{\phi_i \in \Theta_i \\ \phi_i \ne \psi_i}} \lambda_i(\phi_i,\psi_i) - \sum_{\substack{\theta_i \in \Theta_i \\ \theta_i \ne \psi_i}} p_i(\psi_{-i} \mid \theta_i) \lambda_i(\theta_i,\psi_i)
	\\
	= R(\psi).
	\end{multline}

	This condition was first introduced in~\cite{dAGV79} and subsumes the case of conditionally independent priors, i.e., $ p_j(\psi_{-j} \mid \psi_j) = p_j(\psi_{-j} ) $ (refer to~\cite{dAGV79} for an example of priors which are not conditionally independent but still satisfy the condition above).

	The second contribution of this paper is summarized in the following Theorem.
	
	\begin{theorem} \label{thmBSIC}
		If Conditions (A$_{\text{B}}$) and (B) are satisfied then 
		for all $ 0 \le \epsilon < \epsilon_{max} $, there exist taxes $ \big(t_i(\phi)\big)_{\phi \in \Theta,i \in \mN} $ that satisfy BSIC (which implies implementation in BNE) and SBB. 
	\end{theorem}
	\begin{proof}
		Please see Appendix~\ref{app3}.
	\end{proof}

	Stated below are corollaries which successively use stricter conditions.

	\begin{corollary} \label{corolBSIC1}
		If $ \exists $ $ \epsilon_{max} > 0 $ such that for all $ 0 \le \epsilon < \epsilon_{max} $, 
		\begin{equation} \label{EQCRBy1}
		\min_{i \in \mN} \min_{(\theta_i,\phi_i) \in \mF_i} \min_{\theta_{-i} \in \Theta_{-i}} h_i(\theta_i,\phi_i,\theta_{-i},\epsilon) > 0
		\end{equation}
		and Condition (B) is satisfies then for all $ 0 \le \epsilon < \epsilon_{max} $ there exist taxes which satisfy BSIC and SBB.
	\end{corollary}

	\begin{corollary} \label{corolBSIC2}
		If $ \exists $ $ \epsilon_{max} > 0 $ such that for all $ 0 \le \epsilon < \epsilon_{max} $, and $ \forall $ $ i \in \mN $, $ (\theta_i,\phi_i) \in \mF_i $, $ \theta_{-i} \in \Theta_{-i} $,
		\begin{subequations} 
			\begin{gather}
			\label{EQCRBy2}
			K(\epsilon) + \epsilon \: g_i(\theta_i,\phi_i,\theta_{-i},\epsilon) > 0, 
			\\
			where \quad  
			K(\epsilon) \coloneqq \min_{i \in \mN} \min_{(\theta_i,\phi_i) \in \mF_i} \min_{\theta_{-i} \in \Theta_{-i}} K_i(\theta_i,\phi_i,\theta_{-i},\epsilon),
			\\
			\nonumber 
			g_i(\theta_i,\phi_i,\theta_{-i},\epsilon) 
			\\
			\coloneqq \sum_{j \in \mN}  f(v(\hat{x}_\epsilon(\theta_i,\theta_{-i});\theta_j)) - f(v(\hat{x}_\epsilon(\phi_i,\theta_{-i});\theta_j))
			\end{gather}
		\end{subequations} 
		and Condition (B) is satisfied then for all $ 0 \le \epsilon < \epsilon_{max} $ there exist taxes which satisfy BSIC and SBB.
	\end{corollary}

	Since~\eqref{EQCRBy2} $ \Rightarrow $ \eqref{EQCRBy1} $ \Rightarrow $ Condition (A$_{\text{B}}$), hence these Corollaries follow from Theorem~\ref{thmBSIC}.

	%\clearpage 

	\section{Full Bayesian Implementation} \label{secfbi}
	
	In accordance with the majority of the literature on Bayesian implementation (see for e.g.~\cite{borgers2015book,jackson1991bayesian,jackson,hartline2012bayesian}), the main result in the previous section aimed at BSIC condition for implementation. This ensures that truth-telling is a BNE but gives no information about other possible BNE. In general this might be a problem, since when the Bayesian game is played, the designer cannot predict in advance which BNE will be achieved.
	The justification used in such a situation is that of ``focusing''\footnote{This is the same justification as the one used for working without loss of generality, with the Revelation principle and direct mechanisms.}. Along with the mechanism, the selected BNE is also announced by the designer. All the agents are then focused towards this particular BNE and in anticipation that others will be playing according to it, they too choose to play it.
	
	In this section we modify our mechanism by augmenting each agents' message space to include one additional message - which takes values in a continuous space - with the specific aim of achieving full implementation i.e. truth-telling as the only BNE. This will add to the robustness of the mechanism and will put to rest any equilibrium selection issues. Such a modification only requires that the BSIC constraints (see~\eqref{EQbsic}) can be satisfied with strict inequality. This gives the designer room to alter taxes.
	
	For general Bayesian mechanism design, authors in~\cite{mookherjee1990implementation} have derived sufficient conditions under which augmenting the message space results in full implementation. However a specific mechanism is not derived.

	To keep analysis straightforward we make the next assumptions.
	\begin{assumption}
		{Assume private consumption i.e. any agent $ i $'s utility is affected only by level of his/her consumption $ x_i $.}
	\end{assumption}
	\begin{assumption}
		{Assume that $ v(x_i,\theta_i) $ is strictly concave and is differentiable w.r.t. $ x_i $ at all values of $ x_i,\theta_i $.}
	\end{assumption}
	Extending to general public goods is possible, but will make the proofs more technical and obfuscate the basic idea behind them. Differentiability of the utility function is essential to get an explicit modification.

	The new message space, allocation and taxes for any agent $ i $ are 
	\begin{subequations}
		\label{EQnewmech}
		\begin{gather} 
		\mathcal{M}_i = \Theta_i \times [-\delta,+\delta], 
		\quad
		m_i = (\phi_i,y_i).
		\\
		\tilde{x}_{\epsilon,i}(m) = \hat{x}_{\epsilon,i}(\phi) + y_i, 
		\\ 
		\tilde{t}_i(m) = 
		\left\{
		\begin{array}{ll}
		\hat{t}_i(\phi) + y_i v^\prime(\hat{x}_{\epsilon,i}(\phi);\phi_i).  & \mbox{if } y \in (-\delta,+\delta), \\
		B & \mbox{if } y = \pm \delta.
		\end{array}
		\right.
		\end{gather} 
	\end{subequations}
	where $ \hat{x}_{\epsilon}(\phi) $ is the allocation from~\eqref{eqcp} and $ \hat{t}_i(\phi) $ is a tax designed for BSIC and SBB (from the previous section). Also, $ \delta,B > 0 $ are constants chosen by the designer. Here $ \delta $ will be a small constant and how it is chosen will be clear in the proof of Theorem~\ref{thmfbi}. The constant $ B $ is chosen to be large enough so 
	that no rational agent will ever choose message $ y_i = \pm \delta $. This is possible because due to the type sets being discrete, the utilities are bounded. The modification above allows agents to change their allocation by a small amount $ y_i \in [-\delta,+\delta] $. For this increase/decrease in allocation they are charged/subsidized at the ``market'' price $ v^\prime(\hat{x}_{\epsilon,i}(\phi);\phi_i) $ corresponding to the truth-telling strategy. This modification serves the purpose of giving agents more delicate control of their allocation and utility than in a discrete set-up. Technically, this allows the designer to disrupt any BNE at which the price an agent is capable of paying (i.e. his/her derivative without tax terms) doesn't match the market price designed for truthful strategies.

	\begin{theorem} \label{thmfbi}
		Assuming there exists taxes that satisfy BSIC constraints (see for instance~\eqref{EQbsicA}) with strict inequality, for the mechanism defined in~\eqref{EQnewmech} at any true type profile $ \theta = (\theta_i)_{i \in \mN} $,
		\begin{enumerate}
			\item[(1)]  message $ m^\star = (\theta,\underline{0}) $ is a BNE.
		\end{enumerate} 	
		Furthermore, assuming $ \forall~i \in \mN,~(\theta_i,\phi_i) \in \mF_i $, the sign of the quantity $ v^\prime(\hat{x}_{\epsilon,i}(\phi_i,\phi_{-i});\theta_i) - v^\prime(\hat{x}_{\epsilon,i}(\phi_i,\phi_{-i});\phi_i) $ remains same  $ \forall~\phi_{-i} \in \Theta_{-i} $,
		\begin{enumerate}	
			\item[(2)] any message other than $ m^\star $ isn't a BNE. 
		\end{enumerate}
		Hence the mechanism in~\eqref{EQnewmech} achieves full Bayesian implementation.
	\end{theorem} 
	\begin{proof}
		Please see Appendices~\ref{app5} and~\ref{app6}.
	\end{proof}
	
	The no-crossing condition above gives a more definite meaning to private types. In SoU or related problems like here~\eqref{eqcp}, the level of allocation is determined more by the rate of growth of utilities than the value. This is because at optimum, trade-offs between giving infinitesimal allocation to one agent vs. another will be determined by the slope of their utilities at those points. Thus by the above condition, one can implicitly order types - from the one providing smallest slope to the highest. 
	
	A stronger assumption (than above) that works is that the sign of the quantity $ v^\prime(x_i;\theta_i) - v^\prime(x_i;\phi_i) $ is same $ \forall~x_i \in \text{Proj.}(\mX) $ (projection onto the $ i- $th dimension). This might be useful in some cases since verifying this only requires utility functions and type sets whereas $ \epsilon $ and the optimization solution aren't required.
	
	The above assumptions are in fact two of many that can be made to get the result. The proof of Part (2) of Theorem~\ref{thmfbi} requires the quantity
	\begin{equation}
	\sum_{\theta_{-i} \in \Theta_{-i}} p_i(\theta_{-i} \mid \theta_i) \Big[ v^\prime(\hat{x}_i(\sigma(\theta));\theta_i) - 
	v^\prime(\hat{x}_i(\sigma(\theta));\sigma_i(\theta_i)) \Big]
	\end{equation} 
	to remain non-zero whenever $ \sigma_i(\theta_i) \ne \theta_i $. 
	In general, any condition that ensures that for any non-truthful strategy $ \sigma $ there exist an agent $ i $ and type $ \theta_i $ for which the above expression is non-zero, would do as well. For example, given type sets, utilities and the optimization, a condition on the priors might thus be sufficient. Furthermore, such condition wouldn't significantly reduce the set of possible priors since it only requires few equations to remain non-zero.

	Finally, the mechanism presented here has SBB property only at BNE. However, as in proof of Theorem~\ref{thmBSIC}, using the d'AGV form for budget balanced taxes i.e. $ \tilde{t}_i(m) = \tilde{z}_i(m) - \frac{1}{N-1} \sum_{j \ne i} \tilde{z}_j(m) $ a straightforward modification can make the mechanism in~\eqref{EQnewmech} budget balanced off-equilibrium as well.

	\section{Numerical Results\:/\:Examples} \label{secexam}
	
	This section contains numerical examples which have been evaluated to ascertain the scope of application of the existence results provided in the previous sections. In particular we are interested in evaluating the
	gain in overall fairness w.r.t. allocation and taxes, attained with the proposed method.

	Consider $ \Theta_i = S \triangleq \{\theta^H,\theta^L\} ~\forall~ i \in \mN $, where $ \theta^H > \theta^L $ and utilities are quadratic with \emph{private consumption} i.e. $ \forall~i \in \mN $, $ \forall~\theta_i \in S $, $ \forall~x \in \bbR^N_+ $,
	\begin{equation}
	v(x;\theta_i) \coloneqq 2 \theta_i x_i - \theta_i x_i^2.
	\end{equation}
	The constraint set is $ \mX = \{x \in \bbR_+^N \mid \sum_i x_i = 1 \} $ and $ g_\epsilon(z) = z - \epsilon z^2 $. Thus the Centralized optimization problem is
	\begin{equation}
	\label{EQexample}
	\hat{x}_\epsilon(\theta) = \arg\max_{x \in \mX} \sum_{i \in \mN}  \left( 2 \theta_i x_i - \theta_i x_i^2 \right) - \epsilon \left( 2 \theta_i x_i - \theta_i x_i^2 \right)^2.
	\end{equation}
	
	We consider the well known \emph{Gini coefficient} (GC)\footnote{Defined as the ratio of mean of the difference between every possible pair of data points with mean size. The lower the value of GC, the more equitable the allocation. GC of $ 0 $ is perfectly fair and GC of $ 1 $ is absolutely unfair - everyone except an individual receiving $ 0 $. GC is independent of scale, hence can also be used in comparing fairness across different settings.}  as a measure of disparity in allocation. At $ \epsilon =0 $ the objective is exactly the SoU. As $ \epsilon $ starts increasing from $ 0 $ onwards,  the optimization problem is transformed such that higher utilities will be weighed less than lower ones - thereby giving closer to equal distribution of allocation.
	
	For the numerical analysis we consider Bayesian implementation (although Dominant implementation is also possible in this two-type set up) with $ \theta^H = 1 $, $ \theta^L = 0.75 $ and two cases $ N=10,\, 90 $ and vary $ \epsilon $. Figures~\ref{fig:Ginivx},~\ref{fig:Ginivx90} depict the GC of the utility at optimal allocation, $ \{v(\hat{x}_{\epsilon,i}(\theta),\theta_i)\}_{i \in \mN} $ at $ N = 10 $ and $ N = 90 $, respectively. This is done for various type profiles, which due to symmetry can be defined by the number, $ m \in \{0,1,\ldots,N\} $,  of agents with type $ \theta^H $. Also included in the plots is the mean GC where $ m $ is chosen with Binomial($N$,0.5) distribution for $ N = 10 $ and Binomial($ N $,0.1) for $ N = 90 $. Note that in both cases there is a potential for large gains in fairness, as long as a large enough $ \epsilon $ exists, at which implementation is possible.
	
	In general there are two sources of upper bound on $ \epsilon $. One is through the $ \epsilon_{max} $ defined in Condition (A$ _\text{B} $) and other is through the well-defined-ness of the optimization problem~\eqref{EQexample}. For $ \epsilon < \frac{1}{2\theta^H} = 0.5 $, optimization~\eqref{EQexample} is a convex optimization problem and hence has a unique optimizer. One can verify however that the optimization continues to have a unique optimizer even beyond $ 0.5 $ and for all values of $ \epsilon $ within the range depicted in Fig.~\ref{fig:Ginivx},~\ref{fig:Ginivx90}. Therefore we only need to look at Condition (A$ _\text{B} $).  For $ N = 10 $, Condition (A$ _\text{B} $) gives $ \epsilon_{max} \approx 1.504 $ and for $ N = 90 $ we get $ \epsilon_{max} > 5 $. With this we limit our plots to $ \epsilon \le 1.5 $ for $ N = 10 $ and $ \epsilon \le 5 $  for $ N = 90 $.

	In the figures we see that the mean GC can be reduced by $ 42 $ and $ 80 $ percentage points for $ N=10 $ and $ N = 90 $, respectively.

	\begin{figure}
		%\centering
		\includegraphics[scale=0.2]{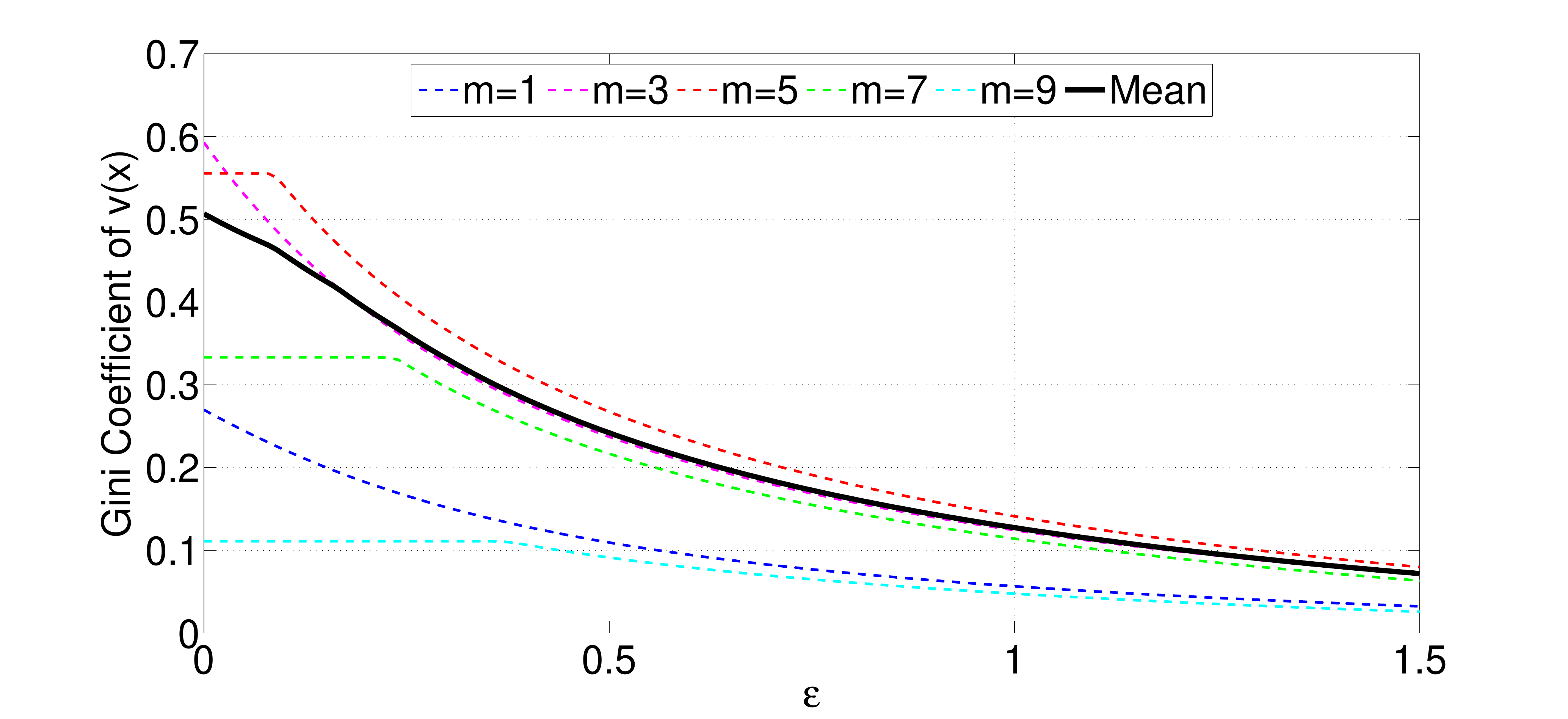}
		\caption{Gini Coefficient of $ \{v(\hat{x}_\epsilon(\theta);\theta_i)\}_{i \in \mN} $ vs $ \epsilon $, $ N = 10 $.}
		\label{fig:Ginivx}
	\end{figure}
	
	\begin{figure}
		%\centering
		\includegraphics[scale=0.2]{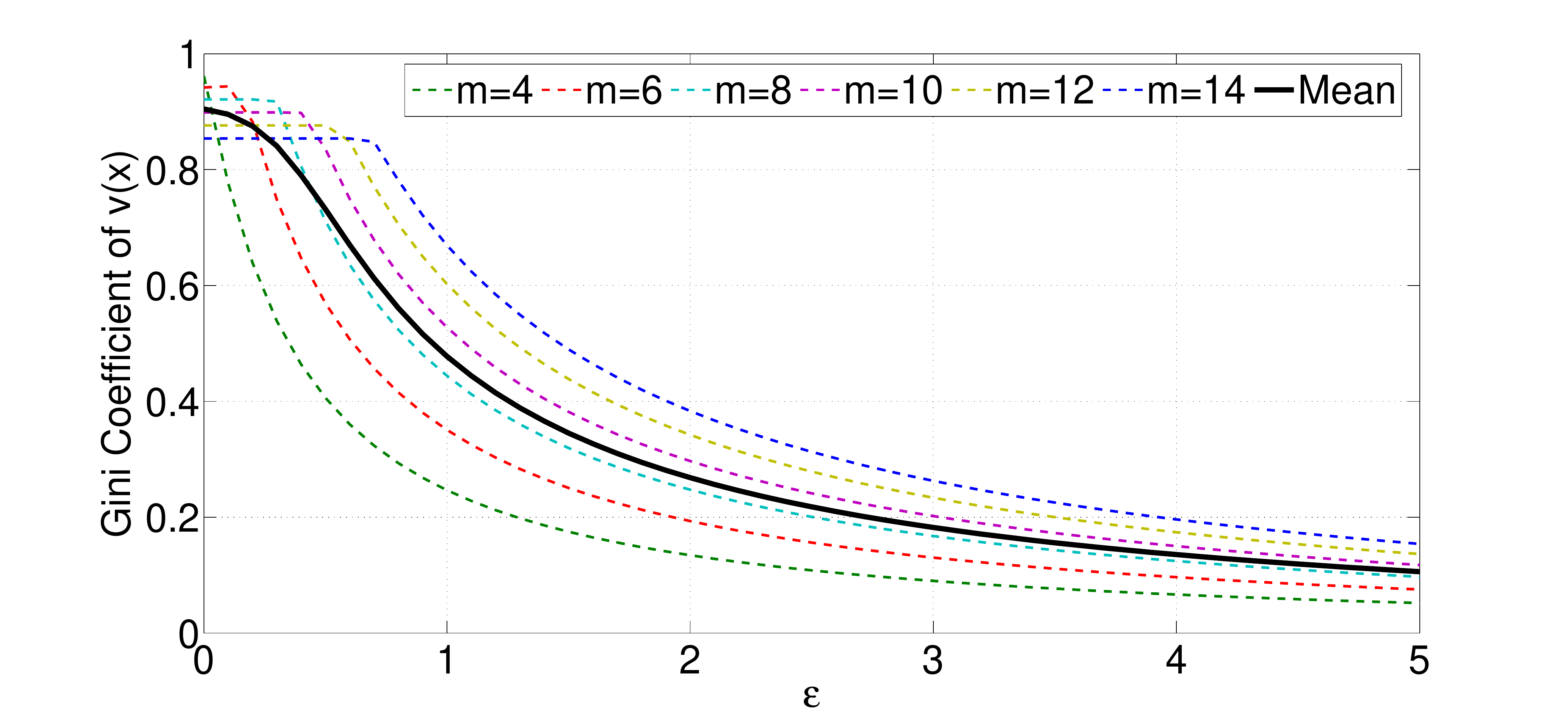}
		\caption{Gini Coefficient of $ \{v(\hat{x}_\epsilon(\theta);\theta_i)\}_{i \in \mN} $ vs $ \epsilon $, $ N = 90 $.}
		\label{fig:Ginivx90}
	\end{figure}

	Finally, Figures~\ref{fig:Taxdistort},~\ref{fig:Taxdistort90} depict the standard deviation in the tax vector $ t(\theta) $ for various type profiles as well as their mean. Here in each case the tax is chosen such that it minimizes the average variance within the feasible space of taxes, as dictated by the BSIC constraints. As $ \epsilon $ increases, for $ N = 10 $ the standard deviation in taxes paid can be driven from $ 0.134 $ to $ 0.0001 $ over the range of permissible $ \epsilon $. As mentioned in the Introduction, this means that tax fluctuation between various agents can be made lower, which can prevent a situation where taxes required to be paid from agents are not within their means. For $ N=90 $, standard deviation reduces by more than one order of magnitude, as it goes from $ 0.046 $ to $ 0.0035 $.

	\begin{figure}
		%\centering
		\includegraphics[scale=0.2]{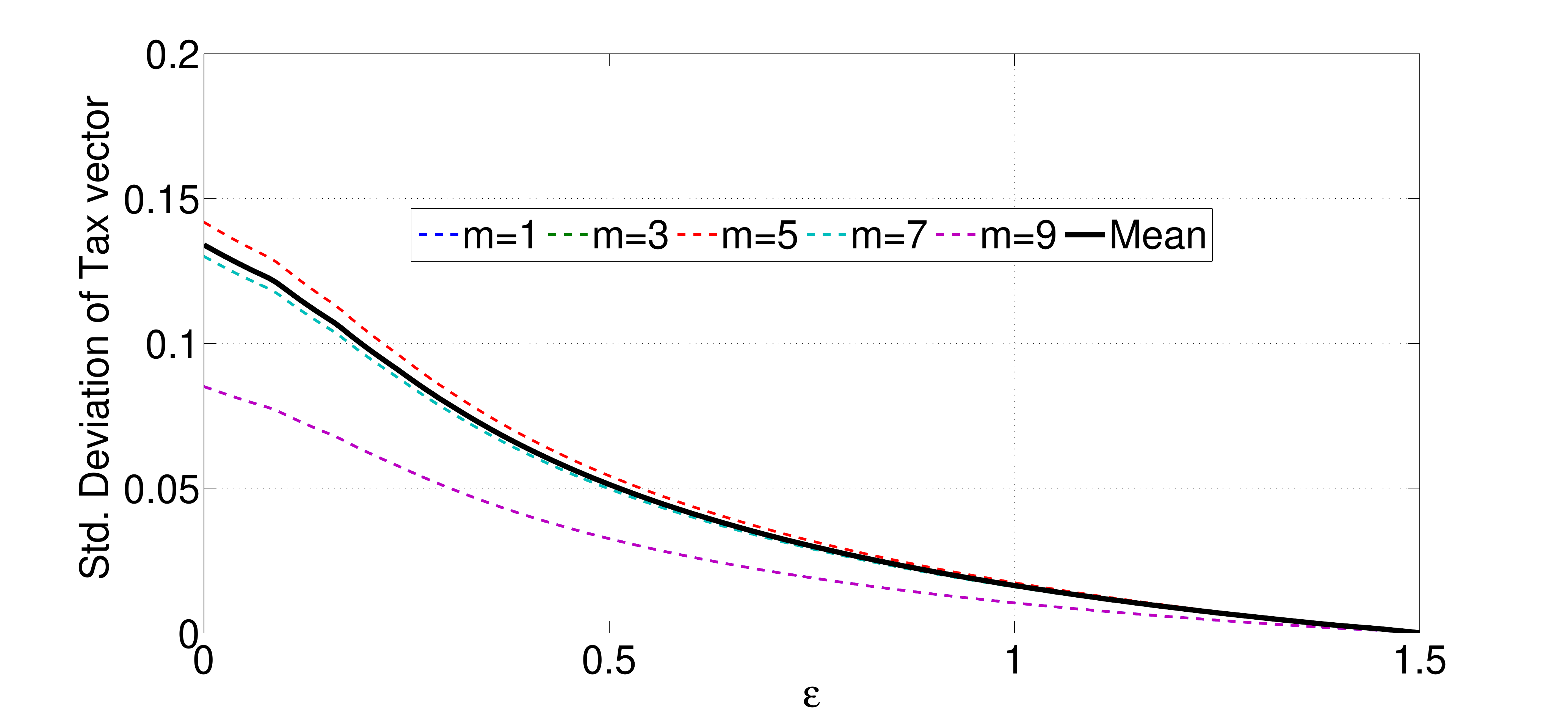}
		\caption{Standard Deviation of $ t(\theta) $ vs $ \epsilon $, $ N = 10 $.}
		\label{fig:Taxdistort}
	\end{figure}
	
	\begin{figure}
		%\centering
		\includegraphics[scale=0.2]{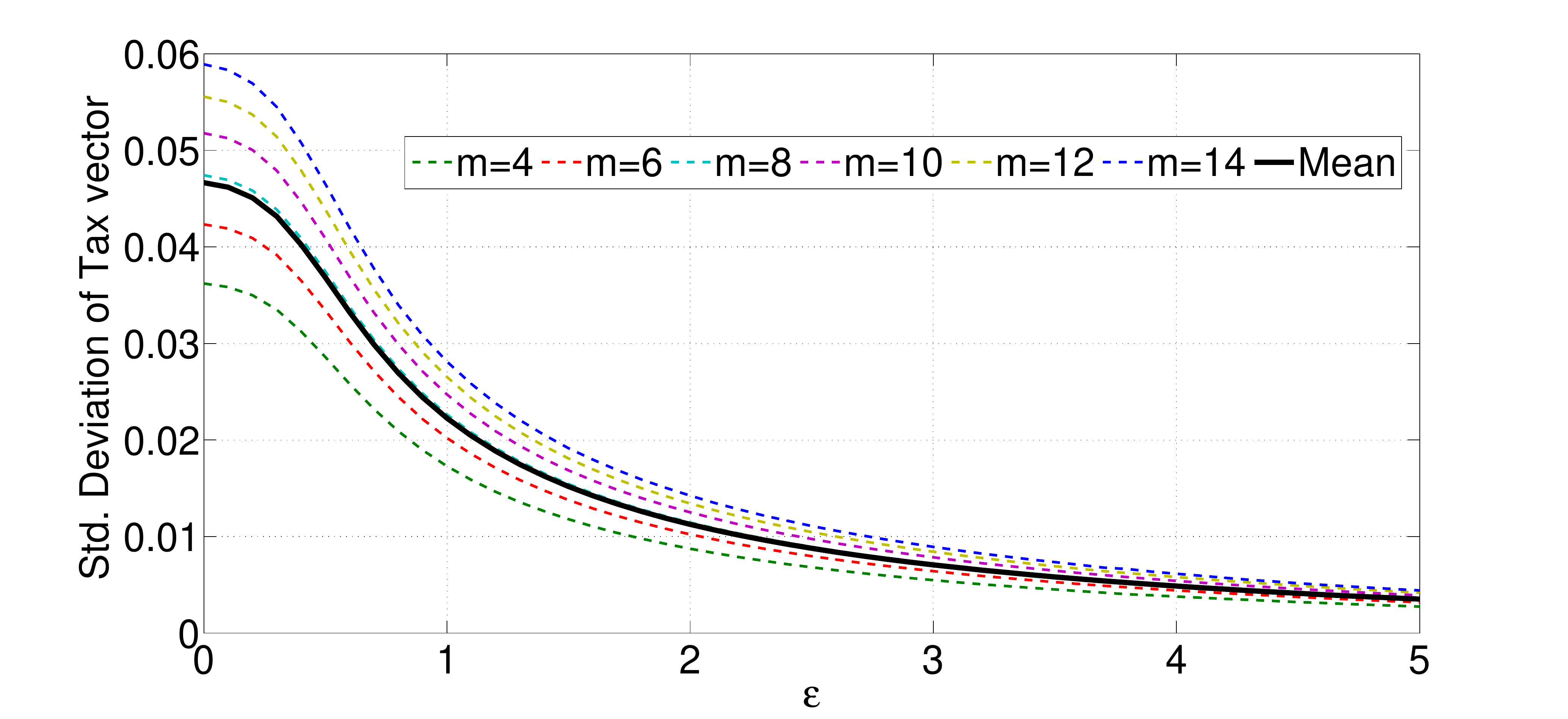}
		\caption{Standard Deviation of $ t(\theta) $ vs $ \epsilon $, $ N = 90 $.}
		\label{fig:Taxdistort90}
	\end{figure}

	\section{Discussion and Conclusions} \label{secdisc}

	This paper introduces a concept of fairness in resource allocation problems pertaining to SoU maximization. The fairness aspect is adjustable (through the selection of the parameter $ \epsilon $ and function $ f(\cdot) $) in a family of functions, thus giving a wide variety of criteria that a designer may choose at their own discretion. The main result in this paper is the proof of existence of mechanisms that implement the fairer allocation in Dominant and Bayesian equilibria (in respective cases). Numerical results indicate that through the proposed techniques there are significant gains in fairness of allocation, within the permissible limits of the design method.
	
	Although the form considered here is $ g_\epsilon(z) = z - \epsilon f(z)  $, it is easy to see that the results can be extended to cases where $ f $ depends on $ \epsilon $; as long as terms of the form $ \epsilon \cdot f $ go to $ 0 $ as $ \epsilon \rightarrow 0 $.
	Ideally one would like to consider the class of $ g_\epsilon(z) = z^{1-\epsilon} $, so as to reconcile with known fair social utilities such as the geometric mean and $ \min $ utility.
	This however may not be a practical necessity since as indicated by the results in Section~\ref{secexam}, even the form $ g_\epsilon(z) = z - \epsilon z^2 $ provides a significant reduction in GC, as well as the standard deviation of taxes. Also, as demonstrated in the Introduction, even theoretically the form $ g_\epsilon(z) = z - \epsilon f(z) $ can be considered as a close approximation to many other families (including the one above) using the Taylor's series.

	Finally a modification of our mechanism was presented, which guarantees truth-telling as the only BNE. This was done by adding one continuous message per agent, other than his/her type. This modification is especially useful in situations where selection of equilibria is too complex to predict.

	\bibliographystyle{IEEEtran}
	\bibliography{IEEEabrv,abhinav}

	\appendices

	\section{Proof of Theorem~\ref{thmDSIC0}} 
	\label{app1}
	
	\begin{proof}
		Note that the taxes are a finite collection of variables, since $ \Theta,\mN $ are both finite sets. For the DSIC constraints to be satisfied, the following constraints must hold $ \forall $ $ i \in \mN $, $ \forall $ $ \theta_{-i} \in \Theta_{-i} $,
		\begin{subequations}
			\begin{multline}
			v(\hat{x}_\epsilon(\theta_i^H,\theta_{-i});\theta_i^H) - t_i(\theta_i^H,\theta_{-i})
			\\
			\ge v(\hat{x}_\epsilon(\theta_i^L,\theta_{-i});\theta_i^H) - t_i(\theta_i^L,\theta_{-i}),
			\label{eqd1} 
			\end{multline} \begin{multline}
			v(\hat{x}_\epsilon(\theta_i^L,\theta_{-i});\theta_i^L) - t_i(\theta_i^L,\theta_{-i})
			\\
			\ge v(\hat{x}_\epsilon(\theta_i^H,\theta_{-i});\theta_i^L) - t_i(\theta_i^H,\theta_{-i}).
			\label{eqd2}
			\end{multline}
		\end{subequations}
		This gives truth-telling as a dominant strategy for agent $ i $ regardless of types of others. For IR, the following constraints must be satisfied $ \forall~i \in \mN $, $ \forall~\theta_{-i} \in \Theta_{-i} $,
		\begin{subequations}
			\begin{gather}
			t_i(\theta_i^H,\theta_{-i}) \le v(\hat{x}_\epsilon(\theta_i^H,\theta_{-i});\theta_i^H),
			\\
			t_i(\theta_i^L,\theta_{-i}) \le v(\hat{x}_\epsilon(\theta_i^L,\theta_{-i});\theta_i^L).
			\end{gather}
		\end{subequations}
		
		From the above sets of constraints, it is clear that one can design $ \big(t_i(\theta_i^H,\theta_{-i}),t_i(\theta_i^L,\theta_{-i})\big) $ separately for each $ i \in \mN $, $ \theta_{-i} \in \Theta_{-i} $. So for any $ i,\theta_{-i} $, the constraints can be rewritten in the form
		\begin{gather}
		\begin{bmatrix}
		1 & -1 \\
		-1 & 1 \\
		1 & 0 \\
		0 & 1
		\end{bmatrix}
		\begin{bmatrix}
		t_1 \\
		t_2
		\end{bmatrix}
		\le
		\begin{bmatrix}
		A_{HH} - A_{LH} \\
		A_{LL} - A_{HL} \\
		A_{HH} \\
		A_{LL}
		\end{bmatrix}
		\end{gather}
		where $ (t_1,t_2) = \big(t_i(\theta_i^H,\theta_{-i}),t_i(\theta_i^L,\theta_{-i})\big) $ and
		\begin{multline}
		A_{HH} = v(\hat{x}_\epsilon(\theta_i^H,\theta_{-i});\theta_i^H),
		~~
		A_{LL} = v(\hat{x}_\epsilon(\theta_i^L,\theta_{-i});\theta_i^L),
		\\
		A_{HL} = v(\hat{x}_\epsilon(\theta_i^H,\theta_{-i});\theta_i^L),
		~~
		A_{LH} = v(\hat{x}_\epsilon(\theta_i^L,\theta_{-i});\theta_i^H).
		\end{multline}
		
		Using the Farkas Lemma, the above system is feasible in $ t $ if and only if $ \forall~\lambda \in \bbR_+^4 $
		\begin{gather}
		\begin{bmatrix}
		1 & -1 \\
		-1 & 1 \\
		1 & 0 \\
		0 & 1
		\end{bmatrix}^\top
		\begin{bmatrix}
		\lambda_1 \\
		\lambda_2 \\
		\lambda_3 \\
		\lambda_4
		\end{bmatrix}
		= 0
		~~\Rightarrow~~
		\begin{bmatrix}
		A_{HH} - A_{LH} \\
		A_{LL} - A_{HL} \\
		A_{HH} \\
		A_{LL}
		\end{bmatrix}^\top
		\begin{bmatrix}
		\lambda_1 \\
		\lambda_2 \\
		\lambda_3 \\
		\lambda_4
		\end{bmatrix}
		\ge 0.
		\end{gather}
		
		The equality constraints on $ \lambda $ give that
		\begin{multline}
		\lambda_1 + \lambda_3 = \lambda_2,
		\quad 
		\lambda_2 + \lambda_4 = \lambda_1
		\\
		~\Rightarrow~
		\lambda_3 + \lambda_4 = 0
		~\Rightarrow~
		\lambda_3 = \lambda_4 = 0.
		\end{multline}
		So $ \lambda \in \bbR_+^4 $ can be parametrized as $ \lambda = (\xi,\xi,0,0)^\top $ for $ \xi \in \bbR_+ $. Thus for feasibility, using Farkas Lemma, the condition that must be satisfied is
		\begin{multline}
		\big(A_{HH}-A_{LH}, A_{LL} - A_{HL}, A_{HH}, A_{LL} \big) \cdot (\xi,\xi,0,0) \ge 0
		\\
		\Leftrightarrow \quad \xi \big(A_{HH} + A_{LL} - A_{HL} - A_{LH}\big) \ge 0
		\\
		\label{EQd10}
		\quad \Leftrightarrow \quad
		A_{HH} + A_{LL} - A_{HL} - A_{LH} \ge 0.
		\end{multline}
	\end{proof}

	\section{Proof of Theorem~\ref{thmDSIC}}
	\label{app2}
	
	\begin{proof}
		From the definition in~\eqref{EQkdef}, we have
		\begin{multline}
		v(\hat{x}_\epsilon(\theta_i,\theta_{-i});\theta_i) - \epsilon f\big(v(\hat{x}_\epsilon(\theta_i,\theta_{-i});\theta_i)\big)
		\\
		+ \sum_{j \ne i} v(\hat{x}_\epsilon(\theta_i,\theta_{-i});\theta_j) - \epsilon f\big(v(\hat{x}_\epsilon(\theta_i,\theta_{-i});\theta_j)\big)
		\\
		=
		v(\hat{x}_\epsilon(\phi_i,\theta_{-i});\theta_i) - \epsilon f\big(v(\hat{x}_\epsilon(\phi_i,\theta_{-i});\theta_i)\big)
		\\
		+ \sum_{j \ne i} v(\hat{x}_\epsilon(\phi_i,\theta_{-i});\theta_j) - \epsilon f\big(v(\hat{x}_\epsilon(\phi_i,\theta_{-i});\theta_j)\big) \\
		{} + K_i(\theta_i,\phi_i,\theta_{-i},\epsilon).
		\end{multline}
		
		Using the above twice, first with $ (\theta_i,\phi_i) = (\theta^H,\theta^L),\theta_{-i} = \theta_{-i} $ and then with $ (\theta_i,\phi_i) = (\theta^L,\theta^H),\theta_{-i} = \theta_{-i} $, and adding the two results in (using the notation from proof of Theorem~\ref{thmDSIC})
		\begin{multline} %\label{eqd3}
		A_{HH} - \epsilon f(A_{HH}) + A_{LL} - \epsilon f(A_{LL}) \\
		= K_i(\theta^H,\theta^L,\theta_{-i},\epsilon) + K_i(\theta^L,\theta^H,\theta_{-i},\epsilon) 
		\\
		{}+{} A_{HL} - \epsilon f(A_{HL}) + A_{LH} - \epsilon f(A_{LH}).
		\end{multline}
		This can be rewritten as 
		\begin{multline} \label{EQd7}
		\Leftrightarrow ~~ A_{HH} + A_{LL} - A_{HL} - A_{LH}
		\\
		=  K_i(\theta^H,\theta^L,\theta_{-i},\epsilon) + K_i(\theta^L,\theta^H,\theta_{-i},\epsilon) 
		\\
		{} + {} \epsilon \big( f(A_{HH}) + f(A_{LL}) - f(A_{HL}) - f(A_{LH}) \big)
		\end{multline}
		
		Thus it is sufficient to prove that the RHS above is non-negative. Owing to Condition (A$_{\text{D}}$), $ \exists $ $ \epsilon_{max} > 0 $ such that the sum of the first two terms in RHS is strictly positive for all $ 0 \le \epsilon < \epsilon_{max} $ and clearly the second term can be made arbitrarily small in magnitude (by choosing a smaller $ \epsilon_{max} $). Hence the condition in~\eqref{EQthm1cond} is satisfied\footnote{Overall the behaviour of $ K_i(\theta_i,\phi_i,\theta_{-i},\epsilon) $ w.r.t. $ \epsilon $ will dictate the value of $ \epsilon_{max},\tilde{\epsilon}_{max} $. This in turn will effect the usefulness of this method, since a designer might want to ensure certain minimum gains in fairness for which he/she might want to choose $ \epsilon $ as large as possible.} for all $ 0 \le \epsilon < \tilde{\epsilon}_{max} $. Here $ \tilde{\epsilon}_{max} $ is bigger or smaller than $ \epsilon_{max} $	depending on whether $ f(A_{HH}) + f(A_{LL}) - f(A_{HL}) - f(A_{LH}) $ is positive or negative in the range $ (0,\epsilon_{max}) $.
	\end{proof}

	\section{Proof of Theorem~\ref{thmBSIC}} 
	\label{app3}
	
	\begin{proof} 
		The utility for any agent $ i $ when other agents are truth-telling is
		\begin{multline}
		\tilde{u}_i(\phi_i \mid \theta_i) = \!\!\! \sum_{\theta_{-i} \in \Theta_{-i}} \!\!\! p_i(\theta_{-i} \mid \theta_i) \cdot 
		\\
		\Big[ v(\hat{x}_\epsilon(\phi_i,\theta_{-i});\theta_i) - t_i(\phi_i,\theta_{-i}) \Big]
		\end{multline}
		where agent $ i $'s true and quoted types are $ \theta_i,\phi_i \in \Theta_i $, respectively. Here we consider taxes in the d'AGV form
		\begin{equation}
		t_i(\psi) = z_i(\psi) - \frac{1}{N-1} \sum_{\substack{j \in \mN \\ j \ne i}} z_j(\psi) \quad \forall~i \in \mN,~\psi \in \Theta.
		\end{equation}
		Taxes in this form always satisfy SBB and any tax function which satisfies SBB can be written in this form. Therefore WLOG, the design variables from here onwards will be $ \{z_j(\psi)\}_{\substack{ j \in \mN \\ \psi \in \Theta }} $.
		
		BSIC constraints can be written as: $ \forall~ i \in \mN,~(\theta_i,\phi_i) \in \mF_i $,
		\begin{subequations}
			\label{EQbsic}
			\begin{gather}
			\label{EQbsicA} 
			\tilde{u}_i(\theta_i \mid \theta_i) \ge \tilde{u}_i(\phi_i \mid \theta_i).
			\\
			\nonumber 
			\Leftrightarrow
			\sum_{\theta_{-i} \in \Theta_{-i}} \!\!\!  p_i(\theta_{-i} \mid \theta_i) \Big[z_i(\theta_i,\theta_{-i}) - \frac{1}{N-1} \sum_{j \ne i} z_j(\theta_i,\theta_{-i})
			\\
			\nonumber 
			-  z_i(\phi_i,\theta_{-i}) + \frac{1}{N-1} \sum_{j \ne i} z_j(\phi_i,\theta_{-i}) \Big]
			\\
			\label{EQby5}
			\le
			\sum_{\theta_{-i} \in \Theta_{-i}}   p_i(\theta_{-i} \mid \theta_i) \big[ v(\hat{x}_\epsilon(\theta_i,\theta_{-i});\theta_i) - v(\hat{x}_\epsilon(\phi_i,\theta_{-i});\theta_i) \big].
			\end{gather}
		\end{subequations}
		
		So the condition for design of taxes is a linear system of inequalities and can be written in the form $ Az \le b $ where the indexing is as follows
		\begin{multline}
		z = \Big( {z}_j(\psi) \Big)  \in \bbR^{D_1},~
		b = \Big( b(i,\theta_i,\phi_i) \Big) \in \bbR^{D_2},
		\\
		A = \Big( A(i,\theta_i,\phi_i \mid j,\psi) \Big) \in \bbR^{D_2 \times D_1}.
		\end{multline}
		for $ D_1 = N \cdot \prod_{j=1}^N L_j $ and $ D_2 = \prod_{i=1}^N (L_i^2-L_i) $.
		The rest of the proof will be to show that this linear system is feasible in variable $ z $, using the \emph{Farkas Lemma}\footnote{ %
			Relevant version of the Farkas alternative result: The system
			\begin{equation}
			x \in \bbR^N, \quad Ax \le b; \quad A \in \bbR^{M\times N}, b \in \bbR^{M}
			\end{equation}
			is feasible {iff} all solutions $ \lambda \in \bbR_+^{M} $ of $ A^\top \lambda = 0 $ satisfy $ b^\top \lambda \ge 0 $.}~\cite[pg.201]{rockafellar1970convex}.

		Consider any $ \lambda = \big(\lambda_i(\theta_i,\phi_i)\big) \in \Lambda \triangleq \{ \lambda_k(\theta_k,\phi_k) \in \bbR_+ \mid k \in \mN, (\theta_k,\phi_k) \in \mF_k \} $ that satisfies $ A^\top \lambda = 0 $, i.e. $ \forall~j \in \mN,~\psi \in \Theta $,
		\begin{multline}
		\sum_{i \in \mN} \sum_{(\theta_i,\phi_i) \in \mF_i} A(i,\theta_i,\phi_i \mid j,\psi) \, \lambda_i(\theta_i,\phi_i) = 0
		\\
		\Leftrightarrow ~~ p_j(\psi_{-j} \mid \psi_j) \sum_{\phi_j \ne \psi_j}  \lambda_j(\psi_j,\phi_j) 
		\\
		{} - \sum_{\theta_j \ne \psi_j} p_j(\psi_{-j} \mid \theta_j) \, \lambda_j(\theta_j,\psi_j)
		\\
		{} - \frac{1}{N-1} \sum_{k \ne j} \Big[ p_k(\psi_{-k} \mid \psi_k) \sum_{\phi_k \ne \psi_k} \lambda_k(\psi_k,\phi_k) \\
		\label{EQlcond}
		{} - \sum_{\theta_k \ne \psi_k} p_k(\psi_{-k} \mid \theta_k) \, \lambda_k(\theta_k,\psi_k) \Big] = 0.
		\end{multline}
		Above equation can be rearranged to give that $ \forall~j \in \mN,~\psi \in \Theta $,
		\begin{multline}
		p_j(\psi_{-j} \mid \psi_j) \sum_{\phi_j \ne \psi_j} \lambda_j(\psi_j,\phi_j) 
		\\
		{} - \sum_{\theta_j \ne \psi_j} p_j(\psi_{-j} \mid \theta_j) \, \lambda_j(\theta_j,\psi_j)
		\\
		= \frac{1}{N} \sum_{k \in \mN} \Big[ p_k(\psi_{-k} \mid \psi_k) \sum_{\phi_k \ne \psi_k} \lambda_k(\psi_k,\phi_k) \\
		\label{EQcondlambda}
		- \sum_{\theta_k \ne \psi_k} p_k(\psi_{-k} \mid \theta_k) \, \lambda_k(\theta_k,\psi_k) \Big].
		\end{multline}
		Denote the RHS above by $ R(\psi) $ - note that it depends only on $ \psi $ and not $ j $.
		
		\begin{lemma}  \label{lemID}
			If $ \lambda \in \Lambda $ satisfies~\eqref{EQcondlambda} then for any $ j \in \mN,~\psi_j \in \Theta_j $,
			\begin{gather}
			\sum_{\substack{\theta_j \in \Theta_j \\ \theta_j \ne \psi_j}} \lambda_j(\theta_j,\psi_j) = \sum_{\substack{ \phi_j \in \Theta_j \\ \phi_j \ne \psi_j }} \lambda_j(\psi_j,\phi_j)
			\end{gather}
		\end{lemma}
		\begin{proof}
			Please see Appendix~\ref{app4}.
		\end{proof}

		With the application of above Lemma, one can rewrite LHS of~\eqref{EQcondlambda} to get that $ \forall~j,\psi $,
		\begin{multline}
		p_j(\psi_{-j} \mid \psi_j) \!\sum_{\phi_j \ne \psi_j}\! \lambda_j(\phi_j,\psi_j) - \sum_{\theta_j \ne \psi_j}\! p_j(\psi_{-j} \mid \theta_j) \, \lambda_j(\theta_j,\psi_j)  \\
		= R(\psi)
		\end{multline}
		
		Condition (B) on priors states that there exist no $ \lambda \in \Lambda $ such that above holds for a non-zero $ R $. Hence $ A^\top \lambda = 0 $ implies that $ R \equiv 0 $, therefore (by~\eqref{EQcondlambda}) $ \forall~j \in \mN $, $ \forall~\psi \in \Theta $,
		\begin{multline}
		\label{EQlmain}
		p_j(\psi_{-j} \mid \psi_j)  \sum_{\phi_j \ne \psi_j}  \lambda_j(\psi_j,\phi_j) 
		\\
		=  \sum_{\theta_j \ne \psi_j}  p_j(\psi_{-j} \mid \theta_j) \, \lambda_j(\theta_j,\psi_j).
		\end{multline}
		
		Next we show that for all $ \lambda \in \Lambda $ that satisfy~\eqref{EQlmain} we have $ b^\top \lambda \ge 0  $ i.e.
		\begin{subequations}
			\begin{gather}
			\label{EQfarkas}
			\sum_{i \in \mN} \sum_{(\theta_i,\phi_i) \in \mF_i} b(i,\theta_i,\psi_i) \, \lambda_i(\theta_i,\phi_i) \ge 0
			\\
			\nonumber 
			\Leftrightarrow \sum_{\substack{i \in \mN, \\ (\theta_i,\phi_i) \in \mF_i}} \!\!\! \lambda_i(\theta_i,\phi_i) \!\! \sum_{\theta_{-i} \in \Theta_{-i}} \!\! p_i(\theta_{-i} \mid \theta_i) \cdot 
			\\
			\Big[ v(\hat{x}_\epsilon(\theta_i,\theta_{-i});\theta_i) - v(\hat{x}_\epsilon(\phi_i,\theta_{-i});\theta_i) \Big] \ge 0.
			\label{EQcond}
			\end{gather}
		\end{subequations}
		Proving this will finish the proof by Farkas Lemma.
		
		For any $ i,\theta_i,\phi_i,\theta_{-i} $, denote $ x = \hat{x}_\epsilon(\theta_i,\theta_{-i}),~\tilde{x} = \hat{x}_\epsilon(\phi_i,\theta_{-i}) $. Rearranging terms from~\eqref{EQkdef}, gives
		\begin{multline}
		v(x;\theta_i) - v(\tilde{x};\theta_i) =
		\\
		\label{EQrhsproof}
		\sum_{\substack{j \in \mN \\ j \ne i}} \Big(v(\tilde{x};\theta_j) - v(x;\theta_j)\Big) + \epsilon \sum_{j \in \mN} \Big(f(v(x;\theta_i)) - f(v(\tilde{x};\theta_i))\Big)
		\\
		{}+{} K_i(\theta_i,\phi_i,\theta_{-i},\epsilon).
		\end{multline}
		
		Denote the RHS expression in~\eqref{EQrhsproof} as $ \eta = \eta_1 + \eta_2 $, where
		\begin{multline}
		\eta_1 = \sum_{\substack{j \in \mN \\ j \ne i}} \Big(v(\tilde{x};\theta_j) - v(x;\theta_j)\Big),
		\\
		\eta_2 = \epsilon \sum_{j \in \mN} \Big(f(v(x;\theta_j)) - f(v(\tilde{x};\theta_j))\Big) + K_i(\theta_i,\phi_i,\theta_{-i},\epsilon).
		\end{multline}
		Now continuing from LHS of~\eqref{EQcond},
		\begin{gather} \label{EQby2}
		\text{LHS of~\eqref{EQcond}} = \!\!\!\sum_{\substack{i\in \mN \\ (\theta_i,\phi_i) \in \mF_i }} \!\!\! \lambda_i(\theta_i,\phi_i) \!\! \sum_{\theta_{-i} \in \Theta_{-i}} \!\! p_i(\theta_{-i} \mid \theta_i) (\eta_1 + \eta_2).
		\end{gather}
		For any fixed $ i $, consider the summation with only $ \eta_1 $ first
		\begin{multline}
		\label{EQby6}
		\sum_{(\theta_i,\phi_i) \in \mF_i} \lambda_i(\theta_i,\phi_i) \sum_{\theta_{-i} \in \Theta_{-i}} p_i(\theta_{-i} \mid \theta_i) \cdot
		\\
		\sum_{\substack{j \in \mN \\ j \ne i}} \Big(v(\hat{x}_\epsilon(\phi_i,\theta_{-i});\theta_j) - v(\hat{x}_\epsilon(\theta_i,\theta_{-i});\theta_j)\Big)
		\\
		= \sum_{\substack{j \in \mN \\ j \ne i}} \sum_{\substack{\phi_i \in \Theta_i, \\ \theta_{-i} \in \Theta_{-i} }} \!\!\! v(\hat{x}_\epsilon(\phi_i,\theta_{-i});\theta_j) \sum_{\theta_i \ne \phi_i} p_i(\theta_{-i} \mid \theta_i) \lambda_i(\theta_i,\phi_i)
		\\
		- \sum_{\substack{j \in \mN \\ j \ne i}} \sum_{\substack{ (\theta_i,\phi_i) \in \mF_i \\ \theta_{-i} \in \Theta_{-i}}} \lambda_i(\theta_i,\phi_i) \: p_i(\theta_{-i} \mid \theta_i) \: v(\hat{x}_\epsilon(\theta_i,\theta_{-i});\theta_j).
		\end{multline}
		By~\eqref{EQlmain}, the inside summation in the first term is equal to $ p_i(\theta_{-i} \mid \phi_i) \sum_{\psi_i} \lambda(\phi_i,\psi_i) $. Incorporating this and changing variables of summation appropriately gives the overall summation from~\eqref{EQby6} equal to $ 0 $.
		Now consider the term in RHS of~\eqref{EQby2} with $ \eta_2 $ 
		\begin{gather}
		\label{EQby7}
		\sum_{\substack{i \in \mN \\ (\theta_i,\phi_i) \in \mF_i}}  \lambda_i(\theta_i,\phi_i) \sum_{\theta_{-i} \in \Theta_{-i}} p_i(\theta_{-i} \mid \theta_i) \cdot \eta_2.
		\end{gather}
		Rearranging terms in $ \eta_2 $, we can write
		\begin{equation}
		\eta_2 = \sum_{j \in \mN} \Big(v(\hat{x}_\epsilon(\theta_i,\theta_{-i});\theta_j) - v(\hat{x}_\epsilon(\phi_i,\theta_{-i});\theta_j)\Big). 
		\end{equation}
		Therefore by Condition (A$_{\text{B}}$), $ \exists $ $ \epsilon_{max} > 0 $ such that for all $ 0 \le \epsilon < \epsilon_{max} $, the inside summation in~\eqref{EQby7} is non-negative $ \forall $ $ i \in \mN $, $ (\theta_i,\phi_i) \in \mF_i $. This finishes the proof by Farkas Lemma, since the expression in~\eqref{EQfarkas} is now shown to be non-negative for all positive $ 0 \le \epsilon < \epsilon_{max} $.
	\end{proof}

	\section{Proof of Lemma~\ref{lemID}}
	\label{app4}
	
	\begin{proof}
		\begin{subequations}
			\begin{align}
			\sum_{\substack{\theta_j \in \Theta_j \\ \theta_j \ne \psi_j}} \lambda_j(\theta_j,\psi_j) - \sum_{\substack{ \phi_j \in \Theta_j \\ \phi_j \ne \psi_j }} \lambda_j(\psi_j,\phi_j)
			\\
			\nonumber
			= \sum_{\psi_{-j}} \Big[ p_j(\psi_{-j} \mid \psi_j) \sum_{\theta_j \ne \psi_j} \lambda_j(\theta_j,\psi_j)
			\\
			- \sum_{\phi_j \ne \psi_j } p_j(\psi_{-j} \mid \phi_j) \, \lambda_j(\psi_j,\phi_j) \Big]
			\\
			\label{EQproof1}
			= \sum_{\psi_{-j}} R(\psi_j,\psi_{-j})
			\\
			\label{EQproof2}
			= \sum_{\psi_{-j}} \Big[ p_k(\psi_{-k} \mid \psi_k) \sum_{\theta_k \ne \psi_k} \lambda_k(\theta_k,\psi_k)
			\\
			\nonumber
			- \sum_{\phi_k \ne \psi_k} p_k(\psi_{-k} \mid \phi_k) \, \lambda_k(\psi_k,\phi_k) \Big],~~k \ne j
			\\
			= 0.
			\end{align}
		\end{subequations}
		Here~\eqref{EQproof1},~\eqref{EQproof2} follows by application of~\eqref{EQcondlambda} and other equations are just by rearranging summation terms.
	\end{proof}

	\section{Proof of  Part (1), Theorem~\ref{thmfbi}}
	\label{app5}

	\begin{proof}
		For strategy
		\begin{gather}
		(\sigma,\rho) = \big\{ \left( \sigma_i:\Theta_i \rightarrow \Theta_i \right),~\left( \rho_i:\Theta_i \rightarrow [-\delta,+\delta] \right) \big\}_{i \in \mN}, 
		\end{gather}
		utility is
		\begin{multline} \label{EQunew}
		u_i(\sigma,\rho \mid \theta_i) = \sum_{\theta_{-i}} p_i(\theta_{-i} \mid \theta_i) \Big[ v\big(\hat{x}_i(\sigma(\theta))+\rho_i(\theta_i);\theta_i\big) 
		\\
		{} - {} \hat{t}_i\big(\sigma(\theta)\big) - \rho_i(\theta_i) v^\prime\big(\hat{x}_{\epsilon,i}(\sigma(\theta));\sigma_i(\theta_i)\big) \Big]
		\end{multline}
		
		First in this proof we establish that for any agent $ i $, true type $ \theta_i $, when other agents use strategy $ m_{-i}^\star = (\theta_{-i},\underline{0}) $, strategy $ \sigma_i(\theta_i) = \theta_i $ gives higher utility than $ \sigma_i(\theta_i) \ne \theta_i $ at any value of $ \rho_i(\theta_i) $. For this start by considering the utility in~\eqref{EQunew} with $ \sigma_{-i}(\theta_{-i}) = \theta_{-i} $ and $ \rho_i(\theta_i) = 0 $, 
		\begin{multline} \label{EQunew2}
		\tilde{u}_i(\sigma_i(\theta_i) \mid \theta_i) = \sum_{\theta_{-i}} p_i(\theta_{-i} \mid \theta_i) \Big[ v\big(\hat{x}_{\epsilon,i}(\sigma_i(\theta_i),\theta_{-i});\theta_i\big) 
		\\
		{} - {} \hat{t}_i\big(\sigma_i(\theta_i),\theta_{-i}\big)  \Big].
		\end{multline}
		From the result in Section~\ref{secbne}, allocation and taxes $ (\hat{x}_{\epsilon},\hat{t}) $ satisfy BSIC condition. Hence the utility in~\eqref{EQunew2} is maximized at $ \sigma_i(\theta_i) = \theta_i $. The minimum difference between the value of such an expression at $ \sigma_i(\theta_i) = \theta_i $ and $ \sigma_i(\theta_i) \ne \theta_i $ is strictly positive since the tax function $ \hat{t} $ satisfies BSIC with strict inequality. 
		Now the utility $ u_i $ in~\eqref{EQunew} clearly depends continuously on $ \rho_i(\theta_i) $, hence by choosing a small enough range for $ \rho_i(\theta_i) $ (i.e. by choosing a small enough $ \delta $) it can be ensured that even in~\eqref{EQunew}, for $ \sigma_{-i}(\theta_{-i}) = \theta_{-i} $, $ u_i $ is maximzed at $ \sigma_i(\theta_i) = \theta_i $.
		
		Now the only thing that remains to be proven is that when other agents quote message $ m_{-i}^\star $ and $ \sigma_i(\theta_i) = \theta_i $, then $ \rho_i(\theta_i) = 0 $ is optimal. 
		For agent $ i $ and true type $ \theta_i $, the utility in~\eqref{EQunew} only depends on $ \rho_i $ through $ \rho_i(\theta_i) $. Optimizing w.r.t. $ \rho_i(\theta_i) $ (for calculating equilibrium) gives
		\begin{multline} \label{EQfbider}
		\frac{\partial u_i}{\partial \rho_i(\theta_i)} = \sum_{\theta_{-i} \in \Theta_{-i}} p_i(\theta_{-i} \mid \theta_i) \Big[ v^\prime(\hat{x}_{\epsilon,i}(\theta)+\rho_i(\theta_i);\theta_i) 
		\\
		{} - {} v^\prime(\hat{x}_{\epsilon,i}(\theta);\sigma_i(\theta_i)) \Big] = 0.
		\end{multline}
		At equilibrium, the above derivative will have to be zero since the end points $ \pm \delta $ cannot be maximizers - since they incur a huge tax $ B $. Also note that since $ v(x_i;\theta_i) $ is assumed to be concave in $ x_i $ (for any $ \theta_i $) the utility $ u_i(\sigma,\rho \vert \theta_i) $ in~\eqref{EQunew} is concave in $ \rho_i(\theta_i) $ and thus the condition in~\eqref{EQfbider} is both necessary and sufficient for optimality w.r.t. $ \rho_i(\theta_i) $.
		
		The expression in~\eqref{EQfbider} is clearly equal to $ 0 $ when $ \sigma_i(\theta_i) = \theta_i $ and $ \rho_i(\theta_i) = 0 $. Hence applying the above result for all $ \theta_i \in \Theta_i $ and $ i \in \mN $ gives that $ m^\star = (\theta,\underline{0}) $ is a BNE.
	\end{proof}

	\section{Proof of Part (2), Theorem~\ref{thmfbi}}
	\label{app6}
	
	\begin{proof} 
		Consider the derivative similar to~\eqref{EQfbider} but with general strategies 
		\begin{multline} \label{EQfbiproof2}
		\frac{\partial u_i}{\partial \rho_i(\theta_i)} =  \sum_{\theta_{-i} \in \Theta_{-i}} p_i(\theta_{-i} \mid \theta_i) \cdot 
		\\
		\Big[ v^\prime(\hat{x}_{\epsilon,i}(\sigma_i(\theta_i),\sigma_{-i}(\theta_{-i})) + \rho_i(\theta_i);\theta_i) 
		\\
		{} - {} v^\prime(\hat{x}_{\epsilon,i}(\sigma_i(\theta_i),\sigma_{-i}(\theta_{-i}));\sigma_i(\theta_i)) \Big].
		\end{multline}
		This being zero is a necessary at BNE. Now we begin by considering the same with $ \rho_i(\theta_i) = 0 $,
		\begin{multline} \label{EQfbiproof2,1}
		\Psi = \sum_{\theta_{-i} \in \Theta_{-i}} p_i(\theta_{-i} \mid \theta_i)
		\Big[ v^\prime(\hat{x}_{\epsilon,i}(\sigma_i(\theta_i),\sigma_{-i}(\theta_{-i}));\theta_i) 
		\\
		{} - {} v^\prime(\hat{x}_{\epsilon,i}(\sigma_i(\theta_i),\sigma_{-i}(\theta_{-i}));\sigma_i(\theta_i)) \Big].
		\end{multline}
		By the assumption in the theorem statement, for $ \sigma_i(\theta_i) \ne \theta_i $, the term inside the brackets above is either positive for all $ \theta_{-i} $ or negative for all $ \theta_{-i} $. Hence the expression $ \Psi $ cannot be zero for any $ i,\theta_i $ when $ \sigma_i(\theta_i) \ne \theta_i $. Since only discrete variables are involved in the expression $ \Psi $, it means that all the values that $ \Psi $ can take (i.e. $ \forall~i,\theta_i,\sigma $ with $ \sigma_i(\theta_i) \ne \theta_i $) are a discrete set (say, $ D $) of points not containing zero. 
		
		The only difference between the expression in~\eqref{EQfbiproof2} and $ \Psi $ is the introduction of the variable $ \rho_i(\theta_i) \in [-\delta,+\delta] $ - on which the expression in~\eqref{EQfbiproof2} depends continuously. This means that the set of all values (say, $ E $) taken by $ \frac{\partial u_i}{\partial \rho_i(\theta_i)} $ is the union of continuous sets centered at the points from set $ D $. Here the range of the continuous sets around each point in $ D $ can be continuously controlled by changing $ \delta $ with the sets becoming a discrete point as $ \delta \rightarrow 0 $. Hence it is clear that one can find $ \delta $ small enough so that even this new set $ E $ doesn't contain zero. Once such a $ \delta $ is chosen it is clear that any strategy $ \sigma $ for which $ \exists~i,\theta_i $ such that $ \sigma_i(\theta_i) \ne \theta_i $ cannot be a BNE since the expression in~\eqref{EQfbiproof2} being zero was a necessary condition for a BNE.
		
		Note that each part of the proof of Theorem~\ref{thmfbi} prescribes a positive upper limit for $ \delta $. The designer will eventually choose a $ \delta $ that satisfies both.
		
		Now consider a strategy of the form $ m_i^\star = (\theta_i,y_i) $ where there is truth-telling but $ y_i \ne 0 $. It is clear that this cannot be a BNE, since by strict concavity of $ v(x_i;\theta_i) $, at $ \sigma_i(\theta_i) = \theta_i $ the expression in~\eqref{EQfbiproof2} goes to zero only for $ \rho_i(\theta_i) = y_i = 0 $.
	\end{proof} 
	
	%\end{spacing}
\end{document}